 \documentclass[reprint, aip, jmp, floatfix, nofootinbib, eqsecnum, citeautoscript, letterpaper, onecolumn]{revtex4-2}

\usepackage[utf8]{inputenc}
\usepackage{tikz}
\usetikzlibrary{shapes,arrows}
\usepackage{amsthm}
\usepackage{bm}
\usepackage{amsmath,amsfonts,amssymb}
\usepackage[letterpaper]{geometry}
\usepackage{color} 
\usepackage{graphicx, xcolor}
\usepackage{hyperref}
\usepackage{algorithm}
\usepackage{algpseudocode}

\usepackage[capitalise]{cleveref}
\crefformat{thm}{Thm.~#2#1#3}
\crefformat{proposition}{Prop.~#2#1#3}
\crefformat{lemma}{Lem.~#2#1#3}
\crefformat{definition}{Defn.~#2#1#3}
\crefformat{equation}{(#2#1#3}
\crefformat{appendix}{App.~#2#1#3}
\crefformat{section}{Sec.~#2#1#3}
\crefformat{algorithm}{Alg.~#2#1#3}

\newcommand{\R}{\mathbb{R}}

\newcommand{\bS}{\mathbb{S}}

\newcommand{\grad}[1]{\nabla #1} 

\newcommand{\hatt}{\hat{t}}
\newcommand{\hn}{\hat{n}}
\newcommand{\hb}{\hat{b}}

\newcommand{\bs}{{\bm{\sigma}}}

\newcommand{\citeInline}[1]{Ref.~[\onlinecite{#1}]}

\theoremstyle{plain}
\newtheorem{thm}{Theorem}[section]
\newtheorem{proposition}[thm]{Proposition}

\theoremstyle{definition}
\newtheorem{definition}[thm]{Definition}

\newtheorem{remark}[thm]{Remark}

\begin{document}
\title{Minimizing Separatrix Crossings through Isoprominence}

 \author{J.W. Burby}
 \affiliation{Los Alamos National Laboratory, Los Alamos, NM 97545 USA}
 \author{N. Duignan}
 \thanks{Email correspondence: Nathan.Duignan@sydney.edu.au}
 \affiliation{Department of Applied Mathematics, University of Colorado, Boulder, CO 80309 USA}
 \affiliation{School of Mathematics and Statistics, University of Sydney, NSW 2050, Australia}
 \author{J.D. Meiss}
 \affiliation{Department of Applied Mathematics, University of Colorado, Boulder, CO 80309 USA}

\begin{abstract}
     A simple property of magnetic fields that minimizes bouncing to passing type transitions of guiding center orbits is defined and discussed. This property, called isoprominence, is explored through the framework of a near-axis expansion. It is shown that isoprominent magnetic fields for a toroidal configuration exist to all orders in a formal expansion about a magnetic axis. Some key geometric features of these fields are described.
\end{abstract}

\date{\today}
\maketitle

\section{Introduction}\label{sec:Intro}
As a collisionless charged particle moves through a strong inhomogeneous magnetic field, $\bm{B}$, its motion comprises three disparate timescales. On the gyrofrequency timescale \cite{Northrop_1963,Cary_Brizard_2009} the particle's position along the field line is frozen while it rapidly rotates around the local magnetic field vector with gyroradius $\rho$. This rapid, nearly-periodic motion gives rise to near conservation of the famous magnetic moment $\mu$, which, more generally, is an adiabatic invariant.\cite{Kruskal_1962,Burby_adiabatic_2020} On a longer timescale, comparable to $L/v$ with $L$ the field scale length and $v$ the characteristic particle speed, the particle moves along magnetic field lines while experiencing the so-called mirror force $ -\mu \bm{b} \cdot \nabla |\bm{B}|$
where $\bm{b}$ is the unit vector along $\bm{B}$.
The perpendicular kinetic energy $\mu|\bm{B}|$, when restricted to a particle's nominal field line plays the role of an effective potential for the particle's parallel dynamics. When the particle's energy is low enough that it is trapped in a well for this potential, it bounces back and forth between a pair of turning points. When the particle is not bouncing---its energy is larger than the highest potential peak---unbounded streaming along the field line ensues. These two scenarios correspond to bouncing and passing orbits, respectively. Finally, on the longest timescale the particle drifts from field line to field line \cite{Littlejohn_bounce_1982} due to the presence of perpendicular gradients in the magnetic field. If the orbit is initially bouncing and this drift does not cause a sudden change in the turning points, then the motion approximately conserves the longitudinal adiabatic invariant $J = \oint u ds$ where $u = \bm{b} \cdot \bm{v}$ is the parallel velocity and $s$ the position along the field line. But when the turning points do suffer a sudden change, then $J$ ceases to be well-conserved. If the particle is still bouncing, the value of $J$ suffers a quasi-random jump, but if one or both of the turning points suddenly disappears $J$ ceases to be well-defined altogether. When a bouncing particle has turning points that suffer any such abrupt change, we say the particle suffers an \textit{orbit-type transition}. It is these orbit-type transitions that bear the responsibility for breakdown of the adiabatic invariance of $J$.

Since orbit-type transitions correlate with deleterious particle transport in magnetic confinement devices such as stellarators,\cite{Helander_2014} the search for magnetic field configurations that minimize the probability of orbit-type transitions warrants detailed investigation. Various strategies have been envisioned for identifying fields which minimize such transitions, including quasisymmetry\cite{Boozer_qs_1983,BKM_2020,Landreman_Paul_2022,Landreman_2022} and omnigeneity\cite{Cary_Shash_1997,Cary_Shash_pop_1997,Parra_2015}. Here we present an initial study of a different strategy that we refer to as \textit{isoprominence}. An isoprominent field is defined as a nowhere-vanishing, divergence-free, vector field $\bm{B}$ such that the height of each potential peak of $\mu\,|\bm{B}|$ is independent of field line, as sketched in Fig.~\ref{fig:IsoPField}, below. More precisely, isoprominence requires that $|\bm{B}|$ is locally constant when restricted to the surface $\Sigma^-$ defined as the set of points such that $\bm{b}\cdot\nabla |\bm{B}|=0$ and $(\bm{b}\cdot\nabla)^2|\bm{B}|<0$. 

As we will discuss in \cref{sec:basic_physics_of_isoprominence}, particles that move in an isoprominent field cannot suffer orbit-type transitions to first order in guiding center perturbation theory. While isoprominent fields do not comprise the most general class of magnetic fields with this property, they enjoy the benefit of an immediately transparent physical interpretation. Moreover, as we show in \cref{sec:FormalExistence}, isoprominent fields may be constructed to all orders in an expansion about a given magnetic axis in powers of distance from the axis. We will present details of this asymptotic expansion as well as examples of magnetic fields that are very nearly isoprominent in \cref{sec:Examples}. Though the fields we construct do not necessarily satisfy the equilibrium conditions of magnetohydrodynamics, we hope that this initial study motivates further study of isoprominence as a potentially useful concept for stellarator optimization.

\begin{figure}[ht]
    \centering
    \includegraphics[width=0.8\linewidth]{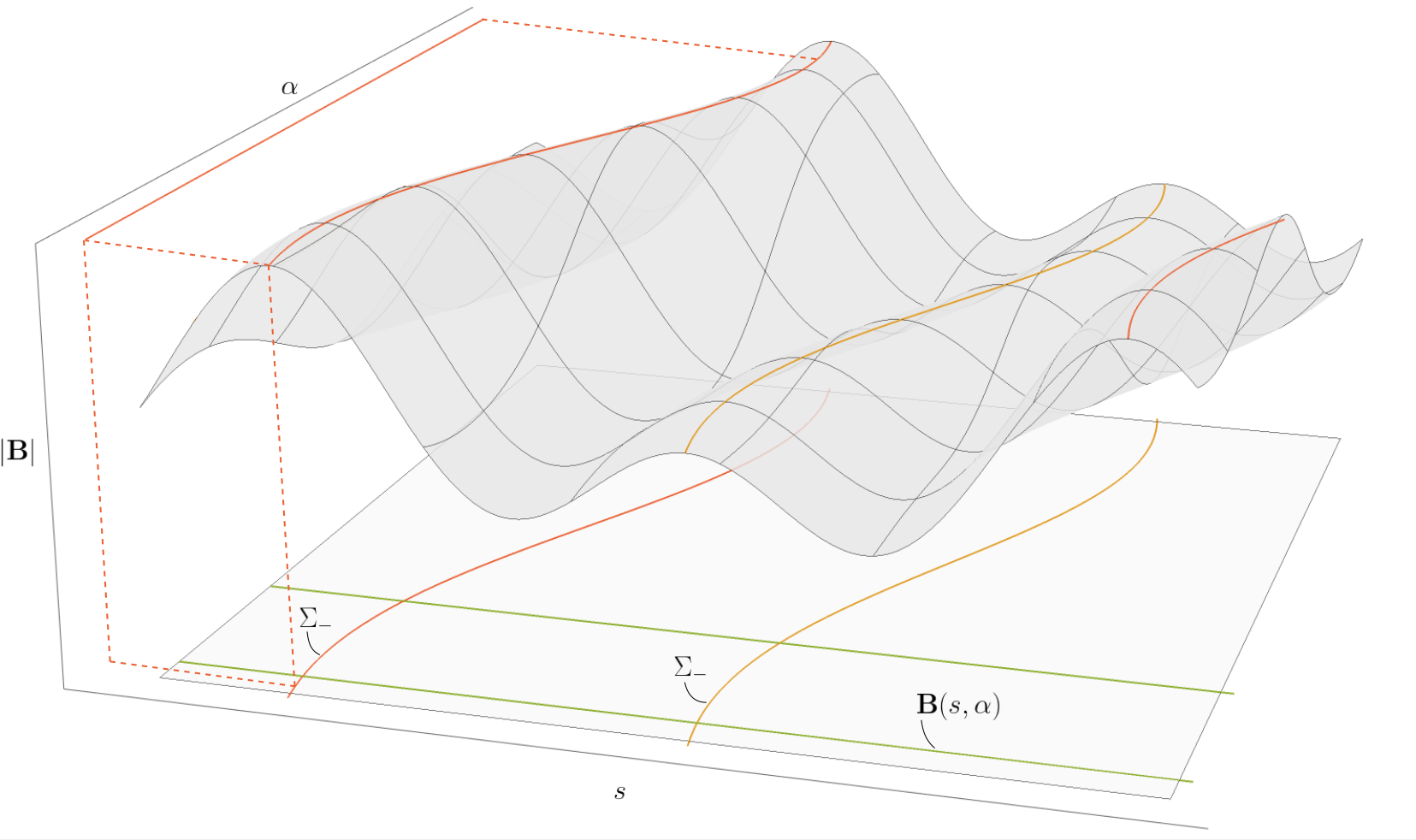}
    \caption{\footnotesize Sketch of the magnetic topography, $|\bm{B}(s,\alpha,0)|$,
    for an isoprominent field, where $s$ denotes field line arc length and $\alpha$ a coordinate on $\Sigma^-$. Note that the maximum values of $|\bm{B}|$ on each of the two components of $\Sigma^-$ does not change from one field line to the next, even though the position of the maxima on the $s$-axis and its value along the valley can both vary.}
    \label{fig:IsoPField}
\end{figure}

\section{Isoprominent magnetic fields\label{sec:basic_physics_of_isoprominence}}

In this section we define the property of isoprominence for a magnetic field. Then, we give an intuitive argument as to why an isoprominent field should mitigate orbit-type transitions in guiding center dynamics. A formal proof is then given in \cref{prop:IsoprominenceMinimizesTransitions}.

Let $\bm{B}$ be a smooth, nowhere-vanishing, divergence-free, vector field defined on an open region $Q\subset \R^3$. The scalar functions 
\begin{align*}
    B & = |\bm{B}|,\quad  B^\prime = \bm{b}\cdot \nabla B,\mbox{ and } B^{\prime\prime}  = (\bm{b}\cdot \nabla)^2B,
\end{align*}
where $\bm{b} = \bm{B}/B$, quantify the magnitude of $\bm{B}$ and its rate of change along integral curves of the magnetic field---called, for simplicity, $\bm{B}$-lines. If $B^\prime$ vanishes at a point $\bs\in Q$ then $B$ restricted to the $\bm{B}$-line passing through $\bs$ has a critical point at $\bs$. We say that $\bs$ is \textbf{critical along }$\bm{B}$. In this case, if $B^{\prime\prime}(\bs)$ is non-zero then $B$ restricted to the $\bm{B}$-line passing through $\bs$ has either a local maximum or local minimum at $\bs$. We say that $B$ is \textbf{locally minimal along }$\bm{B}$ or \textbf{local maximal along }$\bm{B}$ at $\bs$ according to the sign of $B^{\prime\prime}(\bs)$.

\begin{definition}
The \textbf{magnetic ridge} associated with $\bm{B}$ is the smooth submanifold $\Sigma^-\subset Q$ comprising points $\bs\in Q$ such that $B$ is locally maximal along $\bm{B}$ at $\bs$. Note that $\Sigma^-$ may have several connected components.
\end{definition}

Since points $\bs\in\Sigma^-$ are not critical points of $B: Q \to \R$ in the usual sense, the function $B^-=B|_{ \Sigma^-}:\Sigma^-\rightarrow\R$ is generally non-constant, recall Fig.~\ref{fig:IsoPField}. This general situation may be visualized as follows. Fix $\bs\in\Sigma^-$ and let $\ell_\bs(s)$ be the $\bm{B}$-line passing through $\bs$
parameterized by arc length $s$. 
The restriction of $B$ to $\ell_\bs$ defines a single-variable function $B_\bs = B|_{\ell_\bs}$, with a graph $(s,B_\bs(s))$ that 
we call the \textit{magnetic topography} of $\bs$. 
In general this topography has various peaks and valleys. Upon variation of $\bs\in \Sigma^-$ the magnetic topography will continuously deform, the peaks shifting in $s$ and changing in height. Of course, a peak may also collide with a valley, and either peaks or valleys may evaporate. 

With this general picture in mind, isoprominence is defined as follows.
\begin{definition}\label{def:isoprominence}
A nowhere-vanishing, divergence-free vector field $\bm{B}$, defined on an open region $Q\subset \R^3$, is \textbf{isoprominent} if the magnitude of $\bm{B}$ is constant when restricted to a component of the magnetic ridge $\Sigma^-$.
\end{definition}

\noindent For an isoprominent field, the magnetic topography may still deform as $\bs$ varies within $\Sigma^-$, but only in a restricted manner - the peak heights $B_\bs(\bs)$ cannot change. Note also that $\Sigma^-$ for an isoprominent field is a manifold of degenerate critical points for $B$.

Isoprominence mitigates orbit-type transitions for bouncing particles. This can be understood intuitively as follows. In the guiding center approximation, a bouncing particle that starts on a field line $\ell_\bs$
initially oscillates in a magnetic well that is defined by a pair of magnetic peaks 
$s_a^*$ and $s_b^*$, where $\ell_\bs(s_a^*), \ell_\bs(s_b^*) \in \Sigma^-$. 
As $\bs$ varies, so to do $s_a^*$,and $s_b^*$. It follows that $s_a^*$, and $s_b^*$ are functions of $\bs$.
Without loss of generality, let $s_a^*(\bs)<s_b^*(\bs)$. Generally, the
oscillations of a bouncing particle have turning points $s_a(\bs)<s_b(\bs)$ where  ${B}_\bs(s_a(\bs))={B}_\bs(s_b(\bs)) = E/\mu$, for particle energy $E$ and magnetic moment $\mu$. Since the particle bounces in the well, $s_a^*(\bs) < s_a(\bs) < s_b(\bs) < s_b^*(\bs)$ and 
\begin{align}\label{iso_inequality}
    \min\left( B_\bs(s_a^*(\bs)), B_\bs(s_b^*(\bs))\right) > E/\mu.
\end{align} 
Note that equality is not allowed here because such orbits would asymptotically approach $\Sigma^-$ and not bounce.
In order for the particle to suffer an orbit-type transition it must drift onto a field line $\ell_{\bs^\prime}$ where at least one of the bounce points $s_a(\bs^\prime)$  or $s_b(\bs^\prime)$ becomes coincident with $s_a^*(\bs^\prime)$ or $s_b^*(\bs^\prime)$. But if this were to happen for an isoprominent field then, since energy is conserved,
\begin{align*}
    E/\mu = {B}_{\bs^\prime}(s^*_{a/b}(\bs^\prime)) = {B}_\bs(s^*_{a/b}(\bs)),
\end{align*}
which contradicts the inequality \eqref{iso_inequality}.

The weakness of this intuitive reasoning is that it assumes the guiding center Hamiltonian is given precisely by
\begin{equation}\label{eq:H0}
    H_0(\bm{x},u) = \tfrac{1}{2}u^2 + \mu\,B.
\end{equation}
However, in non-constant magnetic fields the single-particle Hamiltonian in guiding center coordinates includes an infinite series of higher-order correction terms.\cite{Cary_Brizard_2009,Burby_Squire_2013,Burby_Squire_2014} To understand the true implications of isoprominence in the context of higher-order guiding center perturbations, we require a more general argument with a slightly weaker conclusion. As we will now explain, the higher-order terms can be accounted for at the price of only approximately ensuring the suppression of type transitions.

Our strategy will amount to first characterizing the set in phase space that separates different bouncing orbit types and then analyzing the component of the guiding center vector field transverse to this separatrix. As we will show, the bounce average of the transverse component vanishes through first order in the guiding center expansion parameter $\epsilon = \rho/L$ for isoprominent fields. This will imply that the bounce-averaged flux of particles across type boundaries is at most $O(\epsilon^2)$.

The various classes of bouncing orbit types are defined in terms of the \textbf{zeroth-order guiding center} (ZGC) equations,
\begin{align*}
    \dot{\bm{x}} &= u\,\bm{b}(\bm{x}) ,\\
    \dot{u} & = -\mu\,\bm{b}(\bm{x})\cdot \nabla B(\bm{x}),
\end{align*}
that describe the motion of a particle with guiding center $\bm{x}\in Q$ and parallel velocity $u\in\R$ in the limit $\epsilon\rightarrow 0$. These equations exactly conserve the energy \eqref{eq:H0}. As mentioned above, solutions of these equations that lie on the boundary of bouncing orbits asymptotically approach $\Sigma^-$ either in the future or the past. If $(\bm{x},u)\in Q\times\R$ is the initial condition for such an orbit and $\bs\in \Sigma^-$ is the magnetic peak it approaches asymptotically, then the value of its conserved energy is given by
\begin{align*}
    H_0(\bm{x},u) = \mu\,B(\bs).
\end{align*}
We say $(\bm{x},u)$ is contained within the \textbf{separatrix}.

The following proposition characterizes the tangent space to the separatrix when the magnetic field $\bm{B}$ is isoprominent.
\begin{proposition}\label{sep_energy_equiv}
Suppose $\bm{B}$ is isoprominent. If $(\bm{x},u)$ lies within the separatrix for the ZGC equations and $\bm{x}\notin \Sigma^-$ then the tangent space to the separatrix at $(\bm{x},u)$ is spanned by vectors of the form
\begin{align*}
    \begin{pmatrix} \delta\bm{x}\\ \delta u \end{pmatrix} = 
    \begin{pmatrix} \delta\bm{x}\\ 
                    -\mu\,\frac{1}{u}\,\delta\bm{x}\cdot \nabla B(\bm{x})
    \end{pmatrix},
    \quad \delta\bm{x}\in\R^3.
\end{align*}
\end{proposition}
\begin{proof}
First we will argue that when $\bm{B}$ is isoprominent, each connected component of the separatrix is contained in a level set of $H_0$. Equivalently, we will show that $H_0$ is locally-constant under the flow of the ZGC flow along the separatrix. Let $(\bm{x},u)$ be a point contained in the separatrix. Without loss of generality, assume $(\bm{x},u)$ has a forward-time limit, and so is a point in the stable manifold of $\Sigma^-$. 
Let $F_t:Q\times\R\rightarrow Q\times\R$ be the ZGC flow. Choose an open neighborhood $U$ contained in the separatrix and containing $(\bm{x},u)$ such that each point $(\bm{x}^\prime,u^\prime)\in U$ has a forward-time limit on a common connected component of $\Sigma^-$. If $(\bm{x}_1,u_1)$ and $(\bm{x}_2,u_2)$ are separatrix points in $U$ then
\begin{align*}
    H_0(\bm{x}_1,u_1) &= \lim_{t\rightarrow\infty}H_0(F_t(\bm{x}_1,u_1)) = \mu\,B(\bm{\sigma}_1,0) ,\\
    H_0(\bm{x}_2,u_2) & = \lim_{t\rightarrow\infty}H_0(F_t(\bm{x}_2,u_2)) = \mu\,B(\bm{\sigma}_2,0),
\end{align*}
where $(\bm{\sigma}_i,0) = \lim_{t\rightarrow\infty}F_t(\bm{x}_i,u_i)$. Since $\bm{\sigma}_1$ and $\bm{\sigma}_2$ are contained in a common connected component of $\Sigma^-$ and $B$ is locally constant on $\Sigma^-$ we must have $B(\bm{\sigma}_1) = B(\bm{\sigma}_2)$. It follows that $H_0(\bm{x}_1,u_1) = H_0(\bm{x}_2,u_2)$, and that $H_0$ is locally constant on the separatrix, as claimed. 

Now we will use the isoenergetic property of the separatrix to deduce the form of its tangent spaces. Let $(\bm{x},u)$ be a point in the separatrix with $u\neq 0 $. Any smooth curve $(\bm{x}(t),u(t))$ contained in the separatrix with $(\bm{x}(0),u(0)) = (\bm{x},u)$ must satisfy $H_0(\bm{x}(t),u(t)) = H_0(\bm{x},u)$. Differentiating in time implies
\begin{align*}
    0& = u\,\dot{u} + \mu\,\dot{\bm{x}}\cdot \nabla B,
\end{align*}
where $\dot{u} = du(0)/dt$ and $\dot{\bm{x}} = d\bm{x}(0)/dt$. This implies $\dot{u} = - \mu\,\dot{\bm{x}}\cdot \nabla B/u$, which is the desired result.

\end{proof}

\begin{remark}
Note that this result merely says the tangent space to the separatrix at $(\bm{x},u)$ is equal to the tangent space to the energy level containing $(\bm{x},u)$ when the magnetic field is isoprominent.
\end{remark}

Next we establish the general result that the leading-order energy for a conservative nearly-periodic system is well-conserved on average.\cite{Kruskal_1962,Burby_adiabatic_2020,Burby_Hirvijoki_2021,Burby_Hirvojoki_Leok_2022}

\begin{proposition}\label{general_averaged_energy}
Let $X_\epsilon = X_0 + \epsilon\,X_1 + \epsilon^2\,X_2 + \dots$ be a formal power series of a vector field in $\epsilon$ on a manifold $M$. Assume that $X_\epsilon$ admits a formal energy invariant $H_\epsilon = H_0 + \epsilon\,H_1 + \epsilon^2\,H_2 + \dots$ and that all trajectories for $X_0$ are periodic with angular frequency function $\omega_0$. Then, after averaging along $X_0$-orbits, the rate of change of $H_0$ along $X_\epsilon$ is $O(\epsilon^2)$.
\end{proposition}
\begin{proof}
For each $z \in M$, let  $z(t)$ be the unique solution to  $\dot{z} = X_0(z)$ with $z(0) =z$.
By assumption this orbit is periodic, so let $T(z) = 2\pi/\omega_0(z)$ denote the period.
The rate of change of $H_0$ along $X_\epsilon$ is 
\begin{align*}
    P & \equiv \frac{d}{dt} {H_0} = \mathcal{L}_{X_\epsilon}H_0 = \mathcal{L}_{X_0}H_0 + \epsilon\,\mathcal{L}_{X_1}H_0 + \epsilon^2\mathcal{L}_{X_2}H_0 +\dots = -\epsilon\,\mathcal{L}_{X_0}H_1 + O(\epsilon^2),
\end{align*}
where we have applied energy conservation $\mathcal{L}_{X_\epsilon}H_\epsilon = 0$.
Denote the average of $P$ along $X_0$-orbits by 
\begin{align*}
    \langle P\rangle(z) = \frac{1}{T(z)}\int_0^{T(z)}P(z(t))\,dt.
\end{align*}
Then the average rate of change of $H_0$ along $X_\epsilon$ is therefore
\begin{align*}
    \langle P\rangle(z) &=- \frac{\epsilon}{T(z)}\int_0^{T(z)}\mathcal{L}_{X_0}H_1(z(t))\,dt + O(\epsilon^2)\\
    &=- \frac{\epsilon}{T(z)}\int_0^{T(z)}\frac{d}{dt}H_1(z(t))\,dt + O(\epsilon^2)\\
     &=- \frac{\epsilon}{T(z)}\bigg(H_1(z(T)) - H_1(z(0))\bigg) + O(\epsilon^2)\\
     & = O(\epsilon^2),
\end{align*}
since $z(T)) = z(0)$.
\end{proof}
\noindent
Finally we apply the previous two observations to establish suppression of type transitions in isoprominent fields.  
\begin{proposition}\label{prop:IsoprominenceMinimizesTransitions}
The bounce-averaged flux of bouncing guiding centers across the ZGC separatrix is $O(\epsilon^2)$ for isoprominent fields.
\end{proposition}
\begin{proof}
The algorithm developed for constructing the guiding center transformation in Ref.\,\onlinecite{Burby_Squire_2013} implies that guiding center coordinates may be chosen so that 
\begin{align}
    \dot{u} & = -\frac{\bm{B}^*}{B_\parallel^*}\cdot \nabla H_\epsilon\label{TGC_u} ,\\
    \dot{\bm{x}} & = \partial_uH_\epsilon\,\frac{\bm{B}^*}{B_\parallel^*} -\epsilon\, \zeta\,\frac{\nabla H_\epsilon\times \bm{b}}{B_\parallel^*},\label{TGC_x}
\end{align}
to all orders in $\epsilon = \rho/L$, where
\begin{align}
    \bm{B}^* & = \bm{B}+\epsilon\,\zeta\,u\,\nabla\times\bm{b} -\epsilon^2\,\zeta\,\mu\,\nabla\times\bm{\mathcal{R}}\label{B_star} ,\\
    B_\parallel^* & = \bm{B}^*\cdot \bm{b}\label{B_par_star},
\end{align}
and $\bm{\mathcal{R}}$ is Littlejohn's gyrogauge vector.\cite{Littlejohn_1984,Littlejohn_1988} The guiding center Hamiltonian $H_\epsilon$ is given by $H_\epsilon = H_0 + \epsilon\,H_1 + \dots$, where the coefficients $H_k$ in general involve high-order derivatives of the magnetic field $\bm{B}$.

By \cref{sep_energy_equiv}, in an isoprominent magnetic field the ZGC separatrix is an energy level for $H_0$. Therefore the distance between a guiding center's phase space location and the ZGC separatrix is proportional to $H_0(z) - H_s$, where $H_s$ denotes the (constant) separatrix energy. By \cref{general_averaged_energy}, the bounce-averaged rate of change of this distance is $O(\epsilon^2)$ for bouncing particles.
\end{proof}

\section{Formal Existence of Isoprominent Fields}\label{sec:FormalExistence}

In this section, the formal existence of isoprominent fields near a \textit{magnetic axis} is established. In this paper, a magnetic axis corresponds to any closed field-line of $\bm{B}$. The only requirement we assert to guarantee the existence of an isoprominent field is that the magnitude of $\bm{B}$ on-axis has only finitely many critical points. We begin by describing useful coordinates in a neighborhood of the axis and then demonstrate that the requirements of isoprominence can be met at each order of a series expansion about the axis.

\subsection{The Near-Axis Framework}    
    
When $r_0:S^1 \to \R^3$ is a closed field-line of $\bm{B}$, convenient, toroidal coordinates in its neighborhood can be obtained usig a \emph{moving frame} such as the Frenet-Serret frame. Such a frame is obtained using the first three derivatives of $r_0(s)$ (see, for instance, \citeInline{Bishop75}). Specifically, assume that the curve $r_0 \in C^3(S^1, \R^3)$ is parameterized by arc length $s$ and define the unit tangent, $\hatt(s) = r_0'$, normal, $\hn(s)$, and bi-normal, $\hb(s)$, vectors. Taking these to be row vectors, they satisfy the matrix \textsc{ode}
\begin{equation}
	\frac{d}{ds} \left(\begin{array}{c}
	\hatt \\ \hn \\ \hb
	\end{array}\right)
	=  \left(\begin{array}{ccc}
        	0 & \kappa(s) & 0 \\
        	-\kappa(s) & 0 & \tau(s) \\
        	0 & -\tau(s) & 0
        	\end{array}\right)\left(\begin{array}{c}
        	\hatt \\ \hn \\ \hb 
    	\end{array}\right),
\end{equation}
where $\kappa(s)$ is the curvature and $\tau(s)$ is the torsion:
\[
	\kappa = |r_0^{\prime\prime}|,\qquad \tau = \frac{(r_0^\prime \times r_0^{\prime\prime})\cdot r_0^{\prime\prime\prime}}{\kappa^2}.
\]
In this paper, we will assume, for simplicity
that $\kappa(s) \neq 0$, so that the normal vector and the torsion are well-defined.
	
The Frenet-Serret frame defines a local embedding
\begin{equation}\label{eq:EuclidMapping}
		\pi_{fs} : D^2\times \bS^1 \to \R^3,\qquad (x,y,s) \mapsto r_0(s)+ x \hn + y \hb.
\end{equation}
In other words, $\pi_{fs}$ is an embedding of the trivial disk bundle $ D^2\times \bS^1$
into a tubular neighborhood of $r_0(s)$ in $\R^3$. 
In these coordinates the metric is
\begin{equation}\label{eq:FSMetric}
\begin{aligned}
		g &= \left(\rho^2 + \tau^2(x^2 + y^2) \right) ds^2 +2\tau (x ds dy -y ds dx) + dx^2 + dy^2, \\
		\rho &\equiv 1-\kappa x.
\end{aligned}
\end{equation}

If $r_0$ is not $C^3$, or, more crucially, if $r_0$ has any inflection points, i.e., points with $\kappa = 0$, then the Frenet-Serret frame does not exist. Even in such cases, however, there are still plenty of choices for an orthonormal frame based on the curve $r_0(s)$. Moreover, with some thought, such a frame can be defined that also has an orthogonal induced metric (in contrast to \eqref{eq:FSMetric}). Such a frame is called \emph{rotation minimizing}. Coordinates based on a rotation-minimizing frame are called \textit{Bishop coordinates}, and these exist when $r_0$ is $C^2$.\cite{Bishop75} Bishop coordinates were used specifically for expansions about magnetic axes in \citeInline{duignanNormalFormsNearaxis2021}.

\subsection{Proof of Formal Existence}\label{sec:ProofOfFormalExistence}

In this section we prove that isoprominent fields exist, at least formally, near a magnetic axis, provided that the magnitude of $\bm{B}$ on-axis has only finitely many critical points. The proof is constructive using the Frenet-Serret coordinates, although, it is not dependent on this particular choice and could be shown with the less restrictive Bishop coordinates. The key idea is to expand in powers of distance from the axis and demonstrate that, at each order, the requirements of isoprominence given in \cref{def:isoprominence} can be satisfied by specifying only the $s$ component of $\bm{B}$ on $\Sigma^-$.
We will then show that $\bm{B}$ can be made divergence-free at each order by specifying either the $x$ or $y$ component of $\bm{B}$. Hence, the formal existence of divergence-free isoprominent fields is guaranteed.

Using Frenet-Serret coordinates $(x,y,s)$, write the contravariant form of magnetic field as
\begin{align}
    \bm{B} &= B^s\partial_s + B^x\partial_x + B^y\partial_y\\
    & = B_0(s)\partial_s + \widetilde{\bm{B}}(x,y,s),
\end{align}
where $B_0(s) = |\bm{B}(s,0,0)|$ is the on-axis field strength and $\widetilde{\bm{B}}$ contains all of the linear and higher-order (in $x,y$) terms of the field. Similarly, write the magnetic field strength as
\begin{align}
    |\bm{B}| = B_0(s) + \widetilde{B}(x,y,s),
\end{align}
where $\widetilde{B}$ contains all of the linear and higher-order (in $x,y$) contributions to the field strength. Assume $B_0(s)$ has only finitely-many critical points $\{s_j\}$ and that $\Sigma^-$ intersects the axis at the points $s_j$ transversely. We will show that isoprominence may be satisfied by specifying $B^s$ on each constant $s$ surface $s = s_j$.

If $s_0 \in \{s_j\}$, then when $\Sigma^-$ is transverse to to the axis at $s_0$, 
it can be represented locally by a graph 
\[
    \Sigma^-_0 := \Sigma^- \cap U_{s_0} = \{(x,y,s)\mid s = \sigma(x,y)\},
\]
where $\sigma(0,0) = s_0$ and $U_{s_0}$ is some neighborhood of $s_0$.


For any function $F$, let $F_k(x,y,s)$ represent the $k^{th}$ term of the power series (in $x,y$) representation of $F$. 
A series representation for the local graph of $\Sigma^-$, 
\[
    \sigma = s_0 + \sigma_1 + \sigma_2 + \dots,
\]
can be obtained by solving the critical value equation
\begin{align}
    0&=\bm{B}\cdot \nabla|\bm{B}|(x,y,\sigma(x,y))\nonumber\\
    & = B^s\,B_0^\prime + B^s\,\partial_s\widetilde{B} + \widetilde{\bm{B}}^\perp\cdot\nabla\widetilde{B} ,\nonumber \\
    \Rightarrow B_0^\prime(\sigma)& = - \partial_s\widetilde{B}(x,y,\sigma) - \frac{\widetilde{\bm{B}}^\perp\cdot\nabla\widetilde{B}}{B^s} (x,y,\sigma), \label{sigma_eqn}
\end{align}
where $\widetilde{\bm{B}}^\perp = B^x\partial_x + B^y\partial_y$. A calculation reveals, 
\begin{align*}
    \left\{ B_0^\prime(\sigma)\right\}_i &= B_0^{\prime\prime}(s_0)\sigma_i + J_{<i} , \\
    \left\{ - \partial_s\widetilde{B}(x,y,\sigma) - \frac{\widetilde{\bm{B}}^\perp\cdot\nabla\widetilde{B}}{B^s} \right\}_i &= -\partial_s\tilde{B}_i(x,y,s_0) - \frac{\bm{B}_1^\perp(x,y,s_0)\cdot \grad \tilde{B}_i(x,y,s_0)}{B_0^s(s_0)} \\
    &\qquad - \frac{\bm{B}_i^\perp(x,y,s_0)\cdot \grad \tilde{B}_1(x,y,s_0)}{B_0^s(s_0)} + F_{<i},
\end{align*}
where $F_{<i}, J_{<i}$ contain $\tilde{B}_j, \sigma_j, \bm{B}^\perp_j$ with orders $j < i$. It follows from these expansions that the $i^{th}$ order terms in \eqref{sigma_eqn} can be solved for $\sigma_i$, given $\tilde{B}_j(s_0), \bm{B}^\perp_j(s_0)$ for $j \leq i$ and their derivatives. Explicit expressions for the first two orders are given in \cref{sec:ExplicitOrder2}.

Given such a series for $\sigma$, we can now establish how the expansions of $\tilde{B}$ and $\bm{B}$ relate. By definition, the two series must obey the algebraic formula $ B_0^s + \tilde{B} = \sqrt{g(\bm{B},\bm{B})}$. Note that the metric \eqref{eq:FSMetric} only contains terms of at most two in $x,y$. Splitting $g$ into these different orders as $g = g_0 + g_1 + g_2$, we obtain at each order,
\begin{equation}\label{eq:BsfromB}
	B_i = \frac{g_0(\bm{B}_i,\bm{B}_0)}{B_0} + G_{<i} =  B^s_i + G_{<i},
\end{equation}
where $G_{<i}$ contains terms of $\bm{B}_j$ for $j<i$. It follows that $B_i$ is equivalent to $B_i^s$ modulo terms in lower-order coefficients. 

Next, we impose the isoprominence condition $B(x,y,\sigma(x,y)) = h$, with $h$ constant. Expanding we have
\[ \left\{ B(x,y,\sigma) \right\}_i = B_i(x,y,s_0) + B_0^\prime(s_0) \sigma_i + H_{< i}, \]
where $H_{<i}$ contains terms in $B_j, \sigma_j$ for $j < i$. By construction, $s_0$ is a critical point of $B_0$, that is, $B_0^\prime(s_0) = 0$. Moreover, we have shown that $\sigma_j$ depends on terms of $\bm{B}_k, B_k$ for $k\leq j$, and that $B_j$ is equivalent to $B_j^s$ modulo terms of lower-order coefficients. It follows that
\begin{equation}\label{eq:isoprominenceExpanded}
	\left\{ B(x,y,\sigma) \right\}_i = B_i^s(x,y,s_0) + H_{< i},
\end{equation}
where $H_{<i}$ contains terms of $\bm{B}_j(x,y,s_0)$ with $j < i$. Hence, we can assert $B_i(x,y,\sigma(x,y)) = 0$ for each $i$ simply by restricting $B^s(x,y,s)$ on $s = s_0$.

The above procedure shows that isoprominence can be satisfied by suitably restricting $B^s(x,y,s)$ on each of $s = s_j$. In the process, the various strata of $\Sigma^-$ local to the magnetic axis, namely $\Sigma_j^- := \Sigma^-\cap U_{s_j}$ for some neighborhood $U_{s_j}$ of $s_j$, are computed. With the suitable restrictions on $B^s(x,y,s)$ made, an isoprominent field can then be obtained by enforcing $B_i^s$ as an interpolation between each restriction on $s = s_j$.

Finally, we show how to ensure the isoprominent field is divergence free. We again proceed by formal expansion, assuming now that $B_s(x,y,s)$ is fixed to ensure isoprominence. 
To impose $\nabla\cdot \bm{B} = 0$ we use the volume element
\begin{equation}\label{eq:Volume}
    \Omega = \rho\,dx\wedge dy\wedge ds .
\end{equation}
The divergence-free condition on $\bm{B}$ becomes 
$\partial_s(\rho B^s) + \partial_x(\rho\,B^x) + \partial_y(\rho\,B^y) = 0$.
Expanding this condition in its power series, we obtain the sequence of constraints
\begin{equation}\label{eq:DivZero}
    \partial_s(\rho\,B^s)_k + \partial_x(\rho\,B^x)_{k+1} + \partial_y(\rho\,B^y)_{k+1} = 0,\quad k = 0,1,2,\dots.
\end{equation}
This is merely a single condition on each of the coefficients of $B^x_{k+1}, B^y_{k+1}$. Hence, at each order we can ensure that $B^s_i$ is chosen for isoprominence, and the single condition on $B^x_{i+1}, B^y_{i+1}$ for incompressibility is achieved.

\section{Examples of Isoprominent Fields}\label{sec:Examples}

Having established the formal existence of isoprominent fields near a magnetic axis, we explore in this section some families of example fields. These families of examples satisfy isoprominence to order three near an axis and are divergence free to all orders. The construction is done in the Frenet-Serret coordinates $(x,y,s)$ of \eqref{eq:EuclidMapping}.

It was demonstrated in \cref{sec:ProofOfFormalExistence}, under the assumption that $\Sigma^-$ intersects the magnetic axis transversely at finitely many points $(0,0,s_j)$, that isoprominence can be satisfied by specifying the $s$ component of $\bm{B}$ on the surfaces $s = s_j$. Moreover, the divergence-free condition to a given order can be be satisfied by specifying some of the coefficients of the $x,y$ components of $\bm{B}$. To create the examples, we begin by first imposing this latter condition before ensuring the isoprominence is satisfied.

\subsection{Divergence-Free condition}\label{sec:divfree}

To impose $\nabla\cdot \bm{B} = 0$ we will proceed as in \cref{sec:ProofOfFormalExistence}, using the the volume element \eqref{eq:Volume}.
This gives the equivalent 
sequence of constraints \eqref{eq:DivZero}
At each order these become
\begin{equation}\label{eq:DivergenceCondition}
\begin{aligned}
    \partial_x B^x_1 + \partial_y B^y_1 &= - B_0^\prime ,\\
    \partial_x B^x_n + \partial_y B^y_n &= -\partial_s\left( B^s_{n-1} - \kappa x B^s_{n-2} \right) + \partial_x \left( \kappa x B^x_{n-1} \right) + \partial_y\left( \kappa x B^y_{n-1} \right),
\end{aligned}
\end{equation}
for $n=2,3,\ldots$.
The first equation is a differential equation for $B^x_1,B^y_1$, and this has the general solution
\begin{align}\label{eq:Bxy1_formula}
    \begin{pmatrix} B^x_1\\ B^y_1 \end{pmatrix} 
        = -\frac{1}{2}B_0^\prime(s)\begin{pmatrix} x\\ y\end{pmatrix} + \frac{1}{2}\begin{pmatrix} \Delta_1(s) & \mu_1(s)\\ \nu_1(s) & -\Delta_1(s)\end{pmatrix}
        \begin{pmatrix}x\\y \end{pmatrix},
\end{align}
where $\Delta_1,\mu_1$, and $\nu_1$ are free periodic functions of $s$. The second equation is a differential equation for $B^x_n,B^y_n$, and using Euler's theorem for homogeneous functions, the solution can be written
\begin{equation}\label{eq:Bxy2_formula}
  \begin{pmatrix}
    B^x_n\\ B^y_n
    \end{pmatrix} = -\frac{1}{n+1}\partial_s[B^s_{n-1} - \kappa\,x\,B_{n-2}^s]\begin{pmatrix}x\\ y \end{pmatrix} + \kappa\,x\,\begin{pmatrix} B^x_{n-1} \\ B^y_{n-1}\end{pmatrix} 
    + \begin{pmatrix} 
        \sum_{j=0}^{n-1} \frac{1}{n-j} \Delta_{n,j}(s) x^{n-j}y^j + \mu_n(s) y^n \\
        \nu_n(s) x^n - \sum_{j=1}^n \frac{1}{j}\Delta_{n,j-1}(s) x^{n-j} y^{j}    
    \end{pmatrix},
\end{equation}
where $\Delta_{n,j}, \mu_n$, and $\nu_n$ are free periodic functions of $s$.

We will require isoprominence only to a fixed order, say $n$. In this case, the expansion \eqref{eq:Bxy2_formula} can be carried out to degree $n$, and we can obtain a field that is still divergence free to all orders of the form
\[
    \bm{B} = \sum_{i=0}^n \bm{B}_i + \bm{B}_{\Delta}.
\]
Here the $\bm{B}_i$ are the terms from \cref{eq:Bxy1_formula}--\cref{eq:Bxy2_formula} which guarantee $\nabla \cdot \bm{B} = 0$ to degree $n-1$ in $x,y$. The field $\bm{B}_{\Delta}$ is homogeneous degree $n+1$ in $x,y$ and is to be determined.

To satisfy the divergence free condition at degree $n$, we require $B_\Delta$ satisfy \eqref{eq:DivergenceCondition}; thus it must be of the form \eqref{eq:Bxy2_formula} with
$n \to n+1$.

With $\bm{B}$ divergence free to degree $n$ in $x,y$, the remaining degree $n+1$ terms---obtained by requiring $B^x_{n+2} = B^y_{n+2}=0$ in \eqref{eq:DivergenceCondition}---are
\begin{equation}\label{eq:BdEqn}
	\partial_x \kappa x B_\Delta^x + \partial_y \kappa x B_\Delta^y = \partial_s( B_\Delta^s - \kappa x B_n^s ).
\end{equation}
To solve this equation we can choose $B_\Delta^s = 0$; this also ensures there are no terms of degree $n+2$ in $x,y$ in the divergence-free equation. Using this ansatz, \eqref{eq:BdEqn} becomes
\begin{align*}
	B_\Delta^x &= - x\left( \partial_x B_\Delta^x + \partial_y B_\Delta^y +\kappa^{-1}\partial_s \left( \kappa  B_n^s\right) \right).
\end{align*}
The free functions $\mu_{n+1},\nu_{n+1},\Delta_{n+1,i}$
drop out of the term $\partial_x B_\Delta^x + \partial_y B_\Delta^y$ as they come from solving $\partial_x B_\Delta^x + \partial_y B_\Delta^y = 0$ at degree $n+1$ in $x,y$.  We can choose $\mu_{n+1}(s) = 0$ and define $\delta := \sum_{j=0}^{n} \frac{1}{n+1-j} \Delta_{n+1,j}(s) x^{n-j}y^j$, 
and recall $B_\Delta^x$ is of the form \cref{eq:Bxy2_formula}, to obtain
\begin{align*}
	 \delta &= - \left(\partial_x B^x_\Delta + \partial_y B_\Delta^y + \kappa^{-1}\partial_s\left(\kappa B^s_n\right) -\frac{1}{n+2}\partial_s\left[ B_n^s - \kappa x B_{n-1}^s \right] + \kappa B_n^x\right), \\
	 	&= \frac{n+3}{n+2} \partial_s\left[ B_n^s - \kappa x B_{n-1}^s \right] - \partial_x\left(\kappa x B_n^x\right) - \partial_y \left( \kappa x B_n^y \right) - \kappa^{-1}\partial_x\left( \kappa B_n^s \right) - \kappa B_n^x .
\end{align*}
It is clear that this can be solved for the coefficients $\Delta_{n+1,j}$ since the right hand side does not depend on these coefficients. Hence the desired $\bm{B}_\Delta$, guaranteeing that $\bm{B}$ is divergence-free to all orders, can be found.

\subsection{First-order terms}\label{sec:firstOrderTerms}

Given any on axis field, $B_0(s)$, it is possible to construct a family of isoprominent fields to first order in the radius for any choice of on-axis rotation number $\iota_0$. Using $\iota_0$ as a parameter in the family ensures that it can be adjusted to study physically relevant phenomena that occur at resonance when $\iota_0$ is rational.

Finding the on-axis rotation number $\iota_0$ can always be done with Floquet theory. However, generally this involves numerical solution of a non-autonomous linear system. Ideally we would like a family of examples where the rotation number is known exactly.

From \eqref{eq:Bxy1_formula}, the field-line dynamics to first order in $x$ and $y$ is
\begin{align*}
	\dot{x} &= \tfrac12\left[ \Delta_1(s) - B_0^\prime(s)\right] x + \tfrac12 \mu_1(s) y ,\\
	\dot{y} &=\tfrac12 \nu_1(s) x -\tfrac12 [\Delta_1(s) + B_0^\prime(s)] y , \\
	\dot{s} &= B_0(s).
\end{align*}

Since $ B_0(s) > 0 $ we can use $s$ as a new time variable. Moreover, we can write
$$
    B_0(s) = B_0(0)e^{2 F(s)}, \quad F(T) = F(0) = 0 ,
$$
so that $F$ represents the on axis variation of the field strength. 
Denoting $d/ds =: {}^\prime$ gives
$$
\left(\begin{array}{c}
	x^\prime \\ y^\prime \end{array}\right) 
	= \left(\begin{array}{cc}
            \tilde{\Delta}_1 - F^\prime & \tilde{\mu}_1  \\
            \tilde{\nu}_1 & -\tilde{\Delta}_1 - F^\prime
        \end{array}\right)\left
        (\begin{array}{c} x \\ y \end{array}\right)  ,
$$
where $\tilde{\Delta}_1 = \frac{\Delta_1}{2 B_0},\, \tilde{\mu}_1 = \frac{\mu_1}{2 B_0},\, \tilde{\nu}_1 = \frac{\nu_1}{2 B_0}$.

We will choose $\tilde{\Delta}_1,\tilde{\mu}_1,\tilde{\nu}_1$ wisely so that the system has a rotation number $\iota_0$. One such choice is
\begin{equation}\label{eq:FreeFunctions}
\begin{aligned}
    \tilde{\mu}_1 &= c_1 e^{2F} ,\\ 
    \tilde{\mu}_1 \tilde{\nu}_1 &= -\left(\tfrac{2\pi\iota_0}{T}\right)^2 +  F^{\prime\prime} - \tilde{\Delta}_1^\prime - (\tilde{\Delta}_1 - F^\prime)^2,
\end{aligned}
\end{equation}
where $c_1$ is a constant.
A computation reveals that the field-line equations become
a single second-order ode:
$$
x^{\prime\prime} = -\left(\tfrac{2\pi\iota_0}{T}\right)^2 x.
$$
It can be concluded that the two eigenfunctions of the linearized system are
$\exp(\pm i \tfrac{2\pi\iota_0}{T} s)$.
Consequently, the Floquet  exponents of the system are simply $\pm i \iota_0$, as desired.

\subsection{Higher-order Isoprominence}
In \cref{sec:ProofOfFormalExistence} it was determined that isoprominence can be met by suitably restricting the $s$ component of $\bm{B}$. To ensure isoprominence is satisfied to some order $n$, the equations of \cref{sec:ProofOfFormalExistence} at each order $j \leq n$ must be satisfied.
This process is recursive: the equation at order $j$ depends on the solution to the equations of lower order. The resulting iterative scheme is sketched in \cref{alg:cap}.

\begin{algorithm}[H]
\caption{Isoprominence iteration scheme}\label{alg:cap}
\begin{algorithmic}
\Require Curvature and torsion of magnetic axis, on-axis field strength $B_0(s)$ with finitely many  maxima.
\For{$ j = 1,2,\dots,n$}
    \For{$ s_i $ maxima of $B_0(s)$}
        \State Compute $B_j(s_i)$ in terms of $B_k^s(s_i)$ for $k\leq j$ using \eqref{eq:BsfromB}
        \State Compute $\sigma_j$ in terms of $B_k^s(s_i)$ for $k\leq j$ using \eqref{sigma_eqn}
    	\State Compute $B_j^s(s_i)$ in terms of $B_k^s(s_i)$ for $k < j$ using \eqref{eq:isoprominenceExpanded}
	\EndFor
	\State Choose a $B^s_j(s)$ so that it interpolates the values specified by each $B^s_j(s_i)$.
\EndFor
\end{algorithmic}
\end{algorithm}

In \cref{sec:ExplicitOrder2} the above algorithm is conducted explicitly to second order giving
\begin{align*}
	B^s_1(x,y,s_i) &= \kappa(s_i)\,x\,B_0(s_i) , \\
	B^s_2(x,y,s_i) &= \kappa^2(s_i)\,x^2\,B_0(s_i) -\frac{1}{2}\tau^2(s_i)(x^2+y^2)\,B_0(s_i)-\frac{1}{2}\left(\frac{(\partial_sB^s_1)^2}{B_0^{\prime\prime}}+\frac{(B^x_1)^2}{B_0}+\frac{(B^y_1)^2}{B_0}\right)\bigg|_{s=s_i} .
\end{align*}
The conditions to arbitrary order can be computed explicitly by the Mathematica notebook available in \citeInline{duignanCodesInvestigatingIsoprominence2022}.

\subsection{Example Fields}

In this section we obtain two simple examples of isoprominent fields as degree-$n$ approximations
about a circular axis. The constructions of these examples follows the procedure:
\begin{enumerate}
    \item Choose a magnetic axis and the on-axis magnetic field, $B_0(s)$. 
    \item Obtain, for each maximum, $s_i$, of $B_0(s)$ the $s$-component $B^s(s_i)$ to satisfy the isoprominence condition to order $n$. 
    \item Interpolate between the values of $B^s(s_i)$ for each $i$ to find  $B^s(s)$.
    \item Obtain the $x,y$ components of the magnetic field using the results in \cref{sec:divfree} to ensure the vector field is divergence free and of polynomial degree $n+1$. 
    \item Choose the free functions $\Delta_{j,k},\mu_k$, and $\nu_k$ arising from the divergence-free condition. 
\end{enumerate} 

In the first example, $B_0(s)$ has only one maxima $s_1$, while the second example has two maxima

\subsubsection*{Example 1}

As a relatively simple family of example fields we choose the following.
\begin{itemize}
	\item[-] The magnetic axis is simply a circle of radius $1/\kappa_0$ so that $\tau(s) = 0$, and $\kappa(s) = \kappa_0$.
	\item[-] The on-axis magnitude is $B_0(s) =  e^{2 (F(s)-F(0))}$ where $F(s) = \tfrac{1}{10} \cos(s)$. Thus it has a single maximum at $s_1 = 0$.
\end{itemize}

The isoprominence condition to order three is used to give $B^s$ at $s = s_1$. To obtain $B^s$ in a neighborhood of the axis we interpolate simply by choosing the degree-$n$
power series terms $B_n^s(s)$ to be independent of $s$, that is, $B_n^s(s) = B_n^s(0)$ for all $s$ and $n \ge 1$.

The first-order terms $B^x_1, B^y_1$ are taken in accordance with \cref{sec:firstOrderTerms} so that the on-axis rotation is $\iota_0$; a parameter of the example. We choose the free function $\Delta_1=0$. 

The free second and third-order terms in \eqref{eq:Bxy1_formula}, resulting from the divergence-free condition, are set to $0$, that is $\Delta_{j,k}=\mu_k =\nu_k = 0$ for $k =2,3$, and all values of $j$.
The fourth-order terms are taken as discussed in \cref{sec:divfree}, to ensure the magnetic field is divergence-free and is a polynomial of degree four. 

The only remaining parameters are the on-axis rotation number $\iota_0$, the free parameter $c_1$ in \eqref{eq:FreeFunctions}, and the curvature of the axis $\kappa_0$. Hence, this is a relatively simple, three-parameter family of divergence-free magnetic fields, that satisfies isoprominence to third order. The formulas for this family are explicitly given in the Mathematica notebook provided in \citeInline{duignanCodesInvestigatingIsoprominence2022}.

Figures~\ref{fig:PoinSecExam1}-\ref{fig:absB_vs_xs_Ex1} show this example for the choice $(\iota_0, c_1, \kappa_0) = (1/\phi,\phi,1)$, where $\phi = \tfrac12(1+\sqrt{5})$ is the golden ratio. A Poincar\'e section of the magnetic field lines is given in \cref{fig:PoinSecExam1}. 
Figure~\ref{fig:Axis_Bsurface} shows the circular magnetic axis, the surface $\Sigma^-$ near the axis, and---using color---the magnitude of $\bm{B}$ on $\Sigma^-$. Notably the deviation of $|\bm{B}|_{\Sigma^-}$ from a constant for $x,y \sim \tfrac12$ is less than $5\%$, which is consistent with the third-order isoprominence of the field. 
Finally, \cref{fig:absB_vs_xs_Ex1} shows $|\bm{B}|$ along orbits that start at $(0,x_0,0)$ for $-0.5 < x_0 < 0.5$. Notice that since field lines are not necessarily in $s$, the magnitude of $\bm{B}$ along a field line will not be periodic in $s$.

\begin{figure}[ht]
    \centering
    \includegraphics[width=0.8\linewidth]{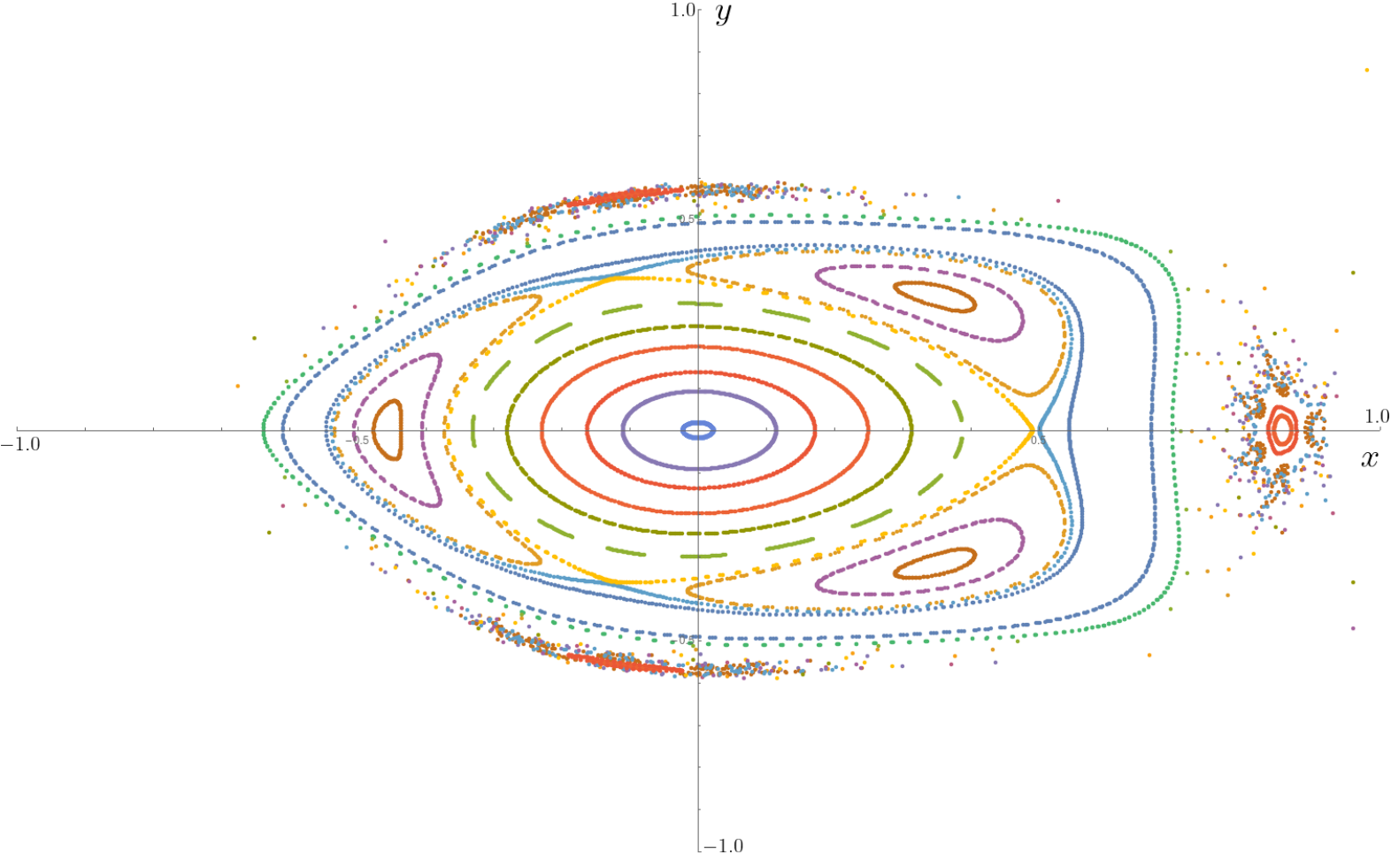}
    \caption{\footnotesize A Poincar\'e section at $s=0$ for the magnetic field lines of example 1.}
    \label{fig:PoinSecExam1}
\end{figure}

\begin{figure}[ht]
    \centering
    \includegraphics[width=0.8\linewidth]{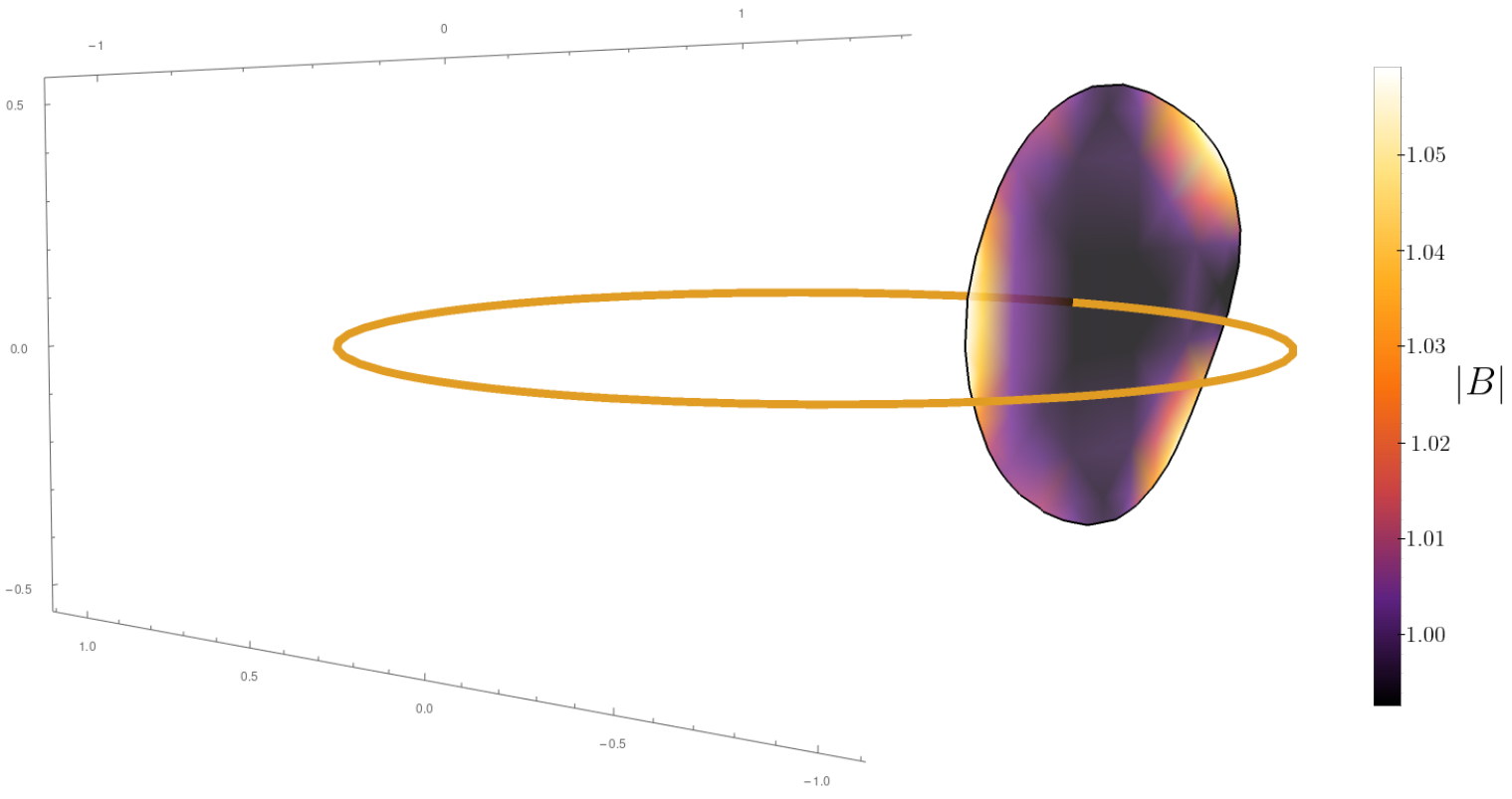}
    \caption{\footnotesize The approximate surface $\Sigma^-$, computed from the isoprominent condition for example 1. A heat map of $|\bm{B}|$ on this surface shows a small variation away from $|\bm{B}|=1$.}
    \label{fig:Axis_Bsurface}
\end{figure}

\begin{figure}[ht]
    \centering
    \includegraphics[width=0.9\linewidth]{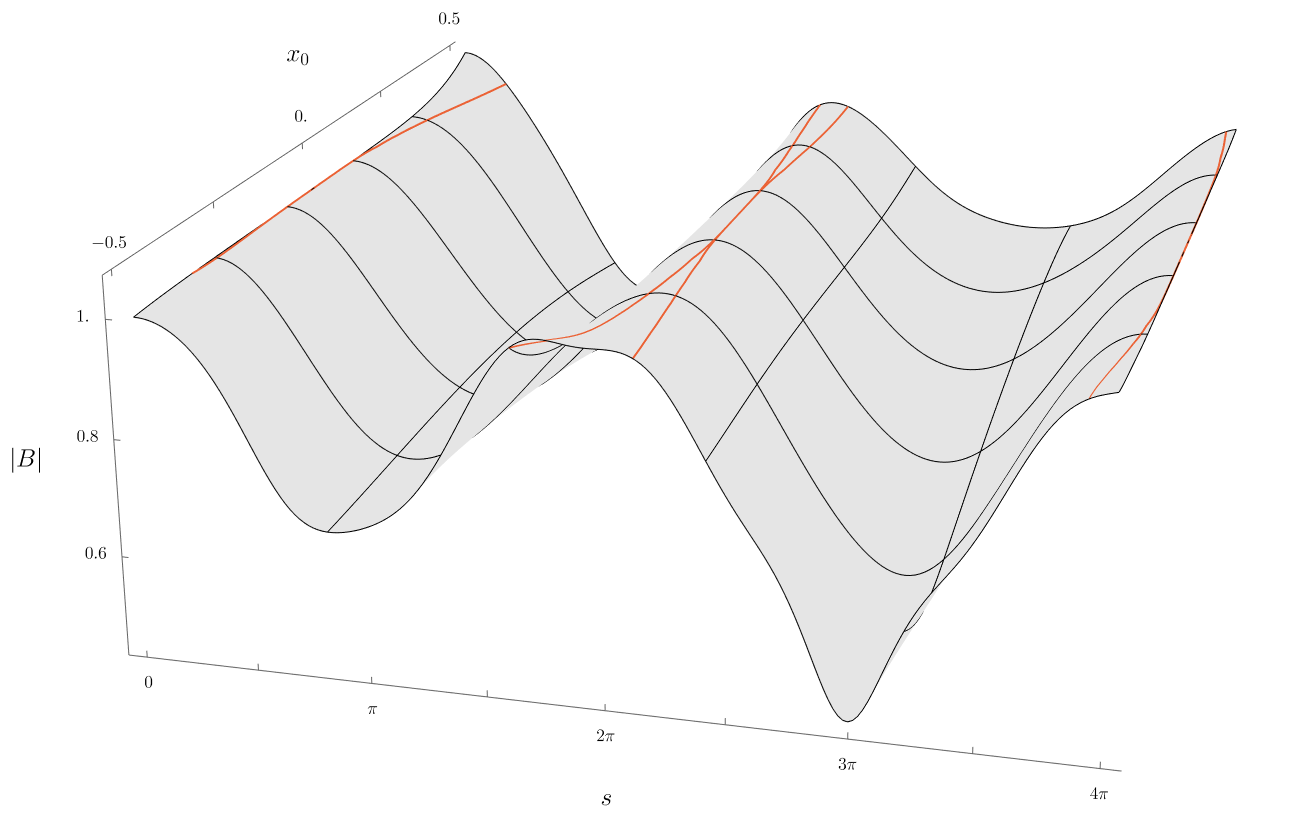}
    \caption{\footnotesize The magnitude of $\bm{B}$ along field lines as a function of the initial condition $(0,x_0,0)$ for example 1. In red are the contours $|\bm{B}| = 1$, the on-axis peak value. 
    These contours split as $x_0$ grows, indicating that there is a region with $|\bm{B}| > 1$ on $\Sigma^-$ due
    the third order truncation of the series. }
    \label{fig:absB_vs_xs_Ex1}
\end{figure}

\subsubsection*{Example 2}

As a second simple family of example fields we start with the following.
\begin{itemize}
	\item[-] The magnetic axis is a circle of radius $1/\kappa_0$ so that $\tau(s) = 0$ and $\kappa(s) = \kappa_0$.
	\item[-] The on-axis magnitude is $B_0(s) =  e^{2 (F(s)-F(0))}$ where $F(s) = a \cos(s) + b\cos(2s) $, with $ 0 < a < 4b $.
\end{itemize}

Since $a < 4b$, the on-axis field $B_0(s)$ has two maxima at $s_1 = 0,\, s_2 = \pi $. Enforcing the isoprominence condition to third order gives the values of $B^s(0,x,y)$ and $B^s(\pi,x,y)$. To obtain $B^s$ for all $s$ we interpolate by taking 
\[ 
    B_n^s(s,x,y) = \tfrac12 B_n^s(0,x,y)(1+\cos(s)) + \tfrac12 B_n^s(\pi,x,y)(1-\cos(s)). 
\]
for $n \ge 1$.

As in the first example, the first-order terms $B^x_1, B^y_1$ are computed using \cref{sec:firstOrderTerms} so that the on-axis rotation is given by the parameter $\iota_0$. The free function $\Delta_1 = 0$ and $\Delta_{j,k} = \mu_k = \nu_k = 0$ for $k =2,3$, and all values of $j$.
The fourth-order terms are computed using \cref{sec:divfree} so that $\bm{B}$ is divergence-free and a polynomial of degree four in $x,y$.

The remaining freedoms are the on-axis rotation number $\iota_0$, $c_1$ \eqref{eq:FreeFunctions}, $\kappa_0$, and the parameters $0 < a < 4b$. Hence, this corresponds to a five-parameter family of divergence-free magnetic fields that satisfies isoprominence to third order.

Figures~\ref{fig:PoinSecExam2}-\ref{fig:absB_vs_xs_Ex2} show this field for  $(\iota_0, c_1, \kappa_0,a,b) = (\tfrac{1}{\phi}, \phi,1, \tfrac1{10}, \tfrac1{10})$, where $\phi = \tfrac12(1+\sqrt{5})$ is the golden ratio. A Poincar\'e section of the magnetic field lines is given in \cref{fig:PoinSecExam2}. Figure~\ref{fig:Axis_Bsurface_Ex2} shows the magnetic axis, the two $\Sigma^-$ surfaces near the axis, and a heat map of the variation in $|\bm{B}|$ on these surfaces. Note that $|\bm{B}|$ varies by less than $7\%$ up to $x,y \sim \tfrac12$, which is consistent with the third-order isoprominence of the field. 
Finally, \cref{fig:absB_vs_xs_Ex2} shows $|\bm{B}|$ along orbits with initial conditions $(0,x_0,0)$, for $-0.5 < x_0 < 0.5$. 

\begin{figure}[ht]
    \centering
    \includegraphics[width=0.8\linewidth]{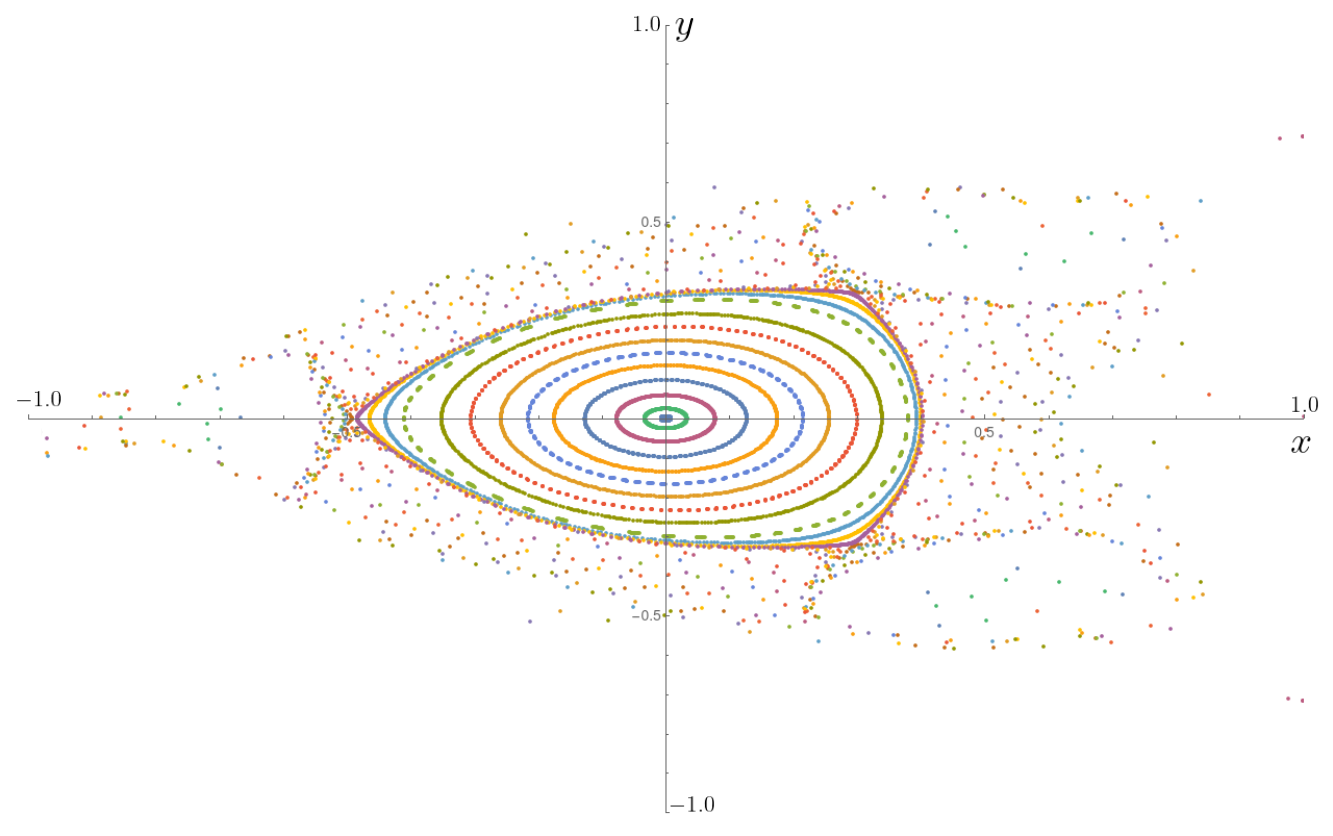}
    \caption{\footnotesize A Poincar\'e section at $s=0$ for the magnetic field lines of example 2.}
    \label{fig:PoinSecExam2}
\end{figure}

\begin{figure}[ht]
    \centering
    \includegraphics[width=0.9\linewidth]{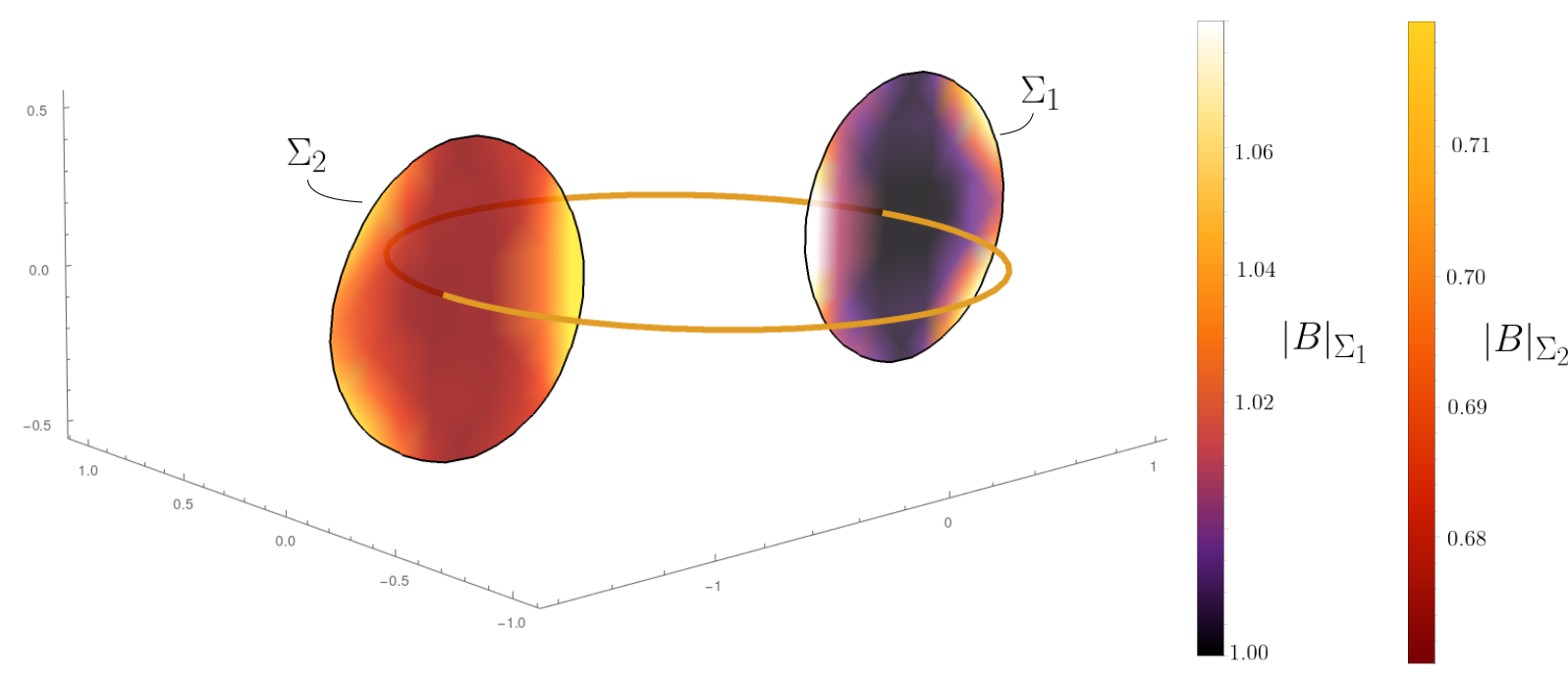}
    \caption{\footnotesize The two strata of $\Sigma^-$ computed from the isoprominent condition for example 2. A heat map of $|\bm{B}|$ is plotted on each surface, showing the error from constant $|\bm{B}|$.}
    \label{fig:Axis_Bsurface_Ex2}
\end{figure}

\begin{figure}
    \centering
    \includegraphics[width=0.9\linewidth]{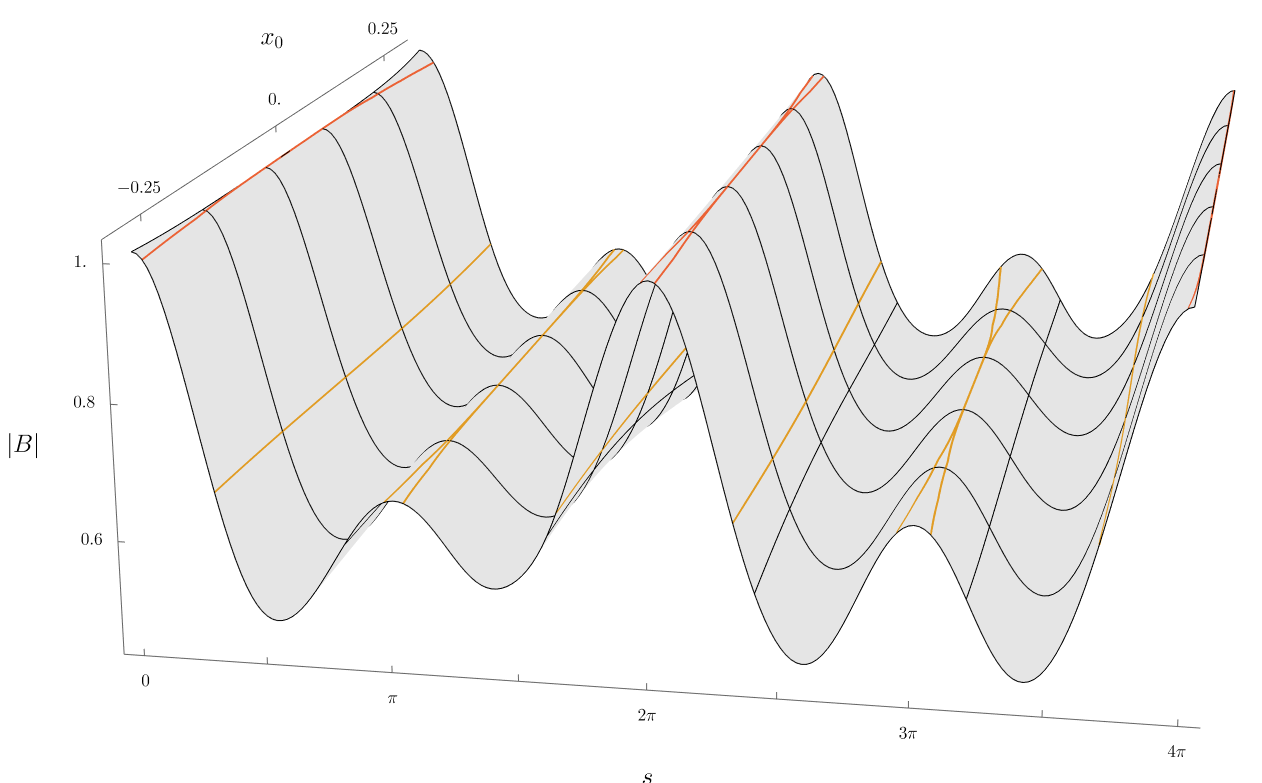}
    \caption{\footnotesize A plot of $|\bm{B}|$ along orbits with initial conditions $(0,x_0,0)$ for example 2. In red are the contours $|\bm{B}| = 1$, the on-axis peak value for $s_1 = 0$ and $2\pi$, and in orange are the $|\bm{B}| = e^{-2/5}$ contours, the desired on-axis peak value at $s_2 = \pi$ and $3\pi$. Both of these contours split   as $x_0$ grows, indicating a deviation from isoprominence due to the third-order truncation.}
    \label{fig:absB_vs_xs_Ex2}
\end{figure}

\section{Discussion}

This paper formulates the concept we call \textit{isoprominence} for a magnetic field: the
magnitude of the field, $B$, should be constant on surfaces where the $\bm{b}$ component of its gradient is locally maximum, recall \cref{def:isoprominence}.
It was shown in \cref{prop:IsoprominenceMinimizesTransitions} that isoprominent magnetic fields minimize separatrix crossings of guiding center orbits. In other words, isoprominent fields minimize orbit-type transitions. In \cref{sec:ProofOfFormalExistence}, power series expansions for isoprominent fields were shown to exist to all orders near a magnetic axis. Two simple example families of nearly-isoprominent fields were obtained and analyzed in \cref{sec:Examples}. 

This paper serves as a first foray into isoprominence. Consequently, there is much left to investigate about this relatively simple property for magnetic fields. The defining condition for isoprominence, namely, that the separatrix of the zeroth-order guiding center motion is a level set of the energy $H$, is simple enough to warrant a search for toroidal configurations with exactly isoprominent fields. Having such an example, rather than the truncated power series in \cref{sec:Examples}, would ease future numerical and theoretical investigations.

Further investigation into realistic physical constraints on the magnetic fields for an isoprominent field is also warranted. For instance, are there isoprominent magnetic fields that also satisfy ideal magneto-hydrostatic equilibrium $J\times B = \nabla p$? Similarly are there isoprominent Beltrami fields $B = \lambda \nabla\times B$, or vacuum fields, $\nabla \times B = 0$? The number of free functions available in the near-axis expansion of isoprominent fields suggests that any of these force-balance constraints can be satisfied. Other constraints worth investigating include those imposed by realistic coil shape design.

Finally,  we note that the current paper contained no investigations of particle trajectories. Even though \cref{prop:IsoprominenceMinimizesTransitions}  minimizes the possibility of separatrix crossings for isoprominent fields in theory, a study of particle trajectories in exact or approximate isoprominent fields would give insight into just how well the crossings are mitigated. This opens up the question of how to accurately and efficiently measure the number of separatrix crossings, so that one can quantify approximate isoprominence. We plan to answer such questions in a future paper.

\section{Acknowledgements}
The work of JWB was supported by the Los Alamos National Laboratory LDRD program under Project No. 20180756PRD4 as well as the US Department of Energy Office of Science as part of the Applied Scientific Computing Research program. ND and JDM were supported by the Simons Foundation under Grant No. 601972, ``Hidden Symmetries and Fusion Energy.''  Helpful discussions with Robert MacKay are gratefully acknowledged.

\appendix

\section{Explicit restriction for order 2 isoprominance}
\label{sec:ExplicitOrder2}
In this appendix we directly construct the conditions for a magnetic field to be isoprominent to second order in the near-axis expansion. Higher-order calculations can be found in the Mathematica codes at \citeInline{duignanCodesInvestigatingIsoprominence2022}. We follow the outline in \cref{alg:cap}. The calculation is performed in Frenet-Serret coordinates \eqref{eq:EuclidMapping}.

The first task is to find a series representation for $\Sigma^-$ near a local maximum point $s_0$ of $B_0(s)$. This local representation is a graph in $x,y$ given by $\sigma = s_0 + \sigma_1 + \sigma_2 + \dots$, where $\sigma_k$ denotes the $k^{\text{th}}$-order term in the power series expansion  in $z = (x,y)$ of $\sigma$ near the magnetic axis $x = y = 0$. The terms in the series are obtained from
expanding the left hand side of \eqref{sigma_eqn} to obtain
\[
   B_0'(\sigma) = B_0''(s_0)(\sigma_1+\sigma_2 + \ldots) + \tfrac12 B_0'''(s_0)(\sigma_1 + \sigma_2 +\ldots)^2 +\ldots
\]
Similarly, the first two terms in the power series expansion of the right hand side of \eqref{sigma_eqn} are
\begin{align*}
    \partial_1 B &\equiv \partial_s B_1 + \frac{\bm{B}^\perp_1\cdot\nabla B_1}{B_0} ,\\
    \partial_2 B & \equiv \partial_s B_2 + \frac{\bm{B}_2^\perp\cdot\nabla B_1}{B_0}+\frac{\bm{B}_1^\perp\cdot\nabla B_2}{B_0} - \frac{\bm{B}_1^\perp\cdot\nabla B_1}{B_0}\frac{B_1}{B_0}.
\end{align*}
Thus the isoprominence condition at first order \eqref{sigma_eqn} requires that
\begin{align}
    \boxed{\sigma_1 = -\left(\frac{\partial_1{B}}{B_0^{\prime\prime}}\right)\bigg|_{s = s_0}}.\label{sigma1_formula}
\end{align}
Similarly, at second order in $z$ we have
\[
    B_0^{\prime\prime}\sigma_2 + \tfrac12 B_0^{\prime\prime\prime}\,\sigma_1^2 = - \partial_s(\partial_1B)\,\sigma_1 - \partial_2B,
\]
or
\begin{equation}
\boxed{ \sigma_2  = -\left(\frac{\partial_2B}{B_0^{\prime\prime}}\right)\bigg|_{s = s_0} + \frac{1}{2}\left(\frac{[\partial_s - B_0^{\prime\prime\prime}/B_0^{\prime\prime}](\partial_1B)^2}{(B_0^{\prime\prime})^2}\right)\bigg|_{s=s_0}}.
\end{equation}

The second step is to relate the power series expansion of $B = |\bm{B}|$ to that of $\bm{B}$. To all orders, the terms in the two series are related through the metric 
\eqref{eq:FSMetric} by $B = \sqrt{g(\bm{B},\bm{B})}$. Expanding both sides in power series and using the fact that $g = g_0 + g_1+g_2$ has only terms up to  $2^{\text{nd}}$-order, we obtain
\[
    B_1 = \frac{1}{2B_0}(2g_0(\bm{B}_0,\bm{B}_1)+ g_1(\bm{B}_0,\bm{B}_0)) ,
\]
at first-order and 
\[
B_2 = \frac{1}{2B_0}(2g_0(\bm{B}_2,\bm{B}_0)+g_0(\bm{B}_1,\bm{B}_1)+ 2g_1(\bm{B}_1,\bm{B}_0)+g_2(\bm{B}_0,\bm{B}_0)) - \frac{1}{8B_0^3}(2g_0(\bm{B}_0,\bm{B}_1)+ g_1(\bm{B}_0,\bm{B}_0))^2,
\]
at second order. These formulas can be simplified by substituting the explicit expression in \eqref{eq:FSMetric}, leading to
\begin{equation}
\boxed{B_1 = B^s_1 - \kappa\,x\,B_0 } , \label{B1_formula}
\end{equation}
\begin{equation}
\boxed{B_2 = B^s_2 - \kappa\,x\,B^s_1 + \frac{1}{2B_0}\left( (B^x_1)^2+(B^y_1)^2 \right) + \frac{1}{2}\tau^2\,(x^2+y^2)B_0 }. \label{B2_formula}
\end{equation}

The third step is to combine the the power series expansions of $\sigma$ and $B$ in to identify an expansion for $B|_{\Sigma^-}$, which is the separatrix energy $h$
of the ZGC equations.
To all orders we have $h(x,y) = B(x,y,\sigma(x,y))$. Expanding both sides of this condition leads to the equations
\begin{align*}
    h_0 & = B_0(s_0),\\
    h_1 & = B_1(x,y,s_0),\\
    h_2 & = B_2(x,y,s_0) + \tfrac12 B^{\prime\prime}_0(s_0)\,\sigma_1^2.
\end{align*}
For isoprominence, it is necessary and sufficient for $h$ to be constant. It follows that we must require that $h_k=0$ for all $k>0$. Using \eqref{B1_formula}, \eqref{B2_formula}, and \eqref{sigma1_formula} these conditions at first and second order may be written
\begin{align*}
    0& = B^s_1(x,y,s_0) - \kappa(s_0)\,x\,B_0(s_0) ,\\
    0& = B^s_2(x,y,s_0) - \kappa(s_0)\,x\,B^s_1(x,y,s_0)+\frac{1}{2B_0(s_0)}
    \left((B^x_1)^2+(B^y_1)^2\right)\bigg|_{s=s_0}+\frac{1}{2}\tau^2(s_0)(x^2+y^2)\,B_0(s_0)\nonumber\\
    &+\frac{1}{2B_0^{\prime\prime}(s_0)}\left( \left[\partial_sB^s_1
    + \frac{B^x_1}{B_0}\partial_xB^s_1 + \frac{B^y_1}{B_0}\partial_yB^s_1\right]
    - \kappa\,B^x_1\right)^2\bigg|_{s=s_0}.
\end{align*}
Consistent with the general argument, each of these conditions may be regarded as specifying $B^s_k(x,y,s_0)$ in terms of lower-degree terms. In particular, for $B^s_1$ and $B^s_2$ we obtain the results
\begin{align}
\boxed{B^s_1(x,y,s_0) = \kappa(s_0)\,x\,B_0(s_0)} ,\label{Bs1_formula}
\end{align}
and
\begin{align}
\boxed{B^s_2(x,y,s_0) =\kappa^2(s_0)\,x^2\,B_0(s_0) -\frac{1}{2}\tau^2(s_0)(x^2+y^2)\,B_0(s_0)-\frac{1}{2}\left(\frac{(\partial_sB^s_1)^2}{B_0^{\prime\prime}}+\frac{(B^x_1)^2}{B_0}+\frac{(B^y_1)^2}{B_0}\right)\bigg|_{s=s_0}} .\label{Bs2_formula}
\end{align}

\newpage
\bibliography{isoProminence}
\bibliographystyle{unsrt}
\end{document}